\documentclass[11pt]{article}
\pdfoutput=1
\usepackage{amsbsy,amsmath,amsthm,amssymb,graphicx,verbatim}
\usepackage{setspace, float, subcaption,xcolor,enumerate}
\usepackage[longnamesfirst,sort]{natbib}
\newcommand{\pcite}[1]{\citeauthor{#1}'s \citeyearpar{#1}}
\usepackage{multirow}

\usepackage{geometry}
\newtheorem{theorem}{Theorem}
\newtheorem{corollary}{Corollary}

\theoremstyle{remark}
\newtheorem{remark}{Remark}

\newtheorem{example}{Example}

\newcommand{\sX}{\mathsf{X}}
\newfont{\msbm}{msbm10 at 11pt}
\newcommand {\R} {\mathbb{R}}

\newcommand {\N} {\mbox{\msbm N}}

\newcommand{\Cov}[0]{\text{Cov}}
\newcommand{\Var}[0]{\text{Var}}

\begin{document}
\onehalfspacing

\title{Assessing and Visualizing Simultaneous Simulation Error} 
%In Monte Carlo Experiments}

\author{Nathan Robertson \\ Department of Statistics \\ University of California, Riverside \\ {\tt nathan.robertson@email.ucr.edu} \and James M. Flegal \\ Department of Statistics \\ University of California, Riverside \\ {\tt jflegal@ucr.edu} \and Dootika Vats \\ Department of Mathematics and Statistics \\ Indian Institute of Technology Kanpur \\ {\tt dootika@iitk.ac.in} \and Galin L. Jones\footnote{Research partially supported by the National Science Foundation.} \\ School of Statistics \\ University of Minnesota \\ {\tt galin@umn.edu} }   

\date{\today}

\maketitle

\begin{abstract} 
  Monte Carlo experiments produce samples in order to estimate
  features of a  given distribution.  However, simultaneous estimation
  of means and quantiles has received little attention, despite being
  common practice. In this setting we establish a multivariate central
  limit theorem for any finite combination of sample means and
  quantiles under the assumption of a strongly mixing process, which
  includes the standard Monte Carlo and Markov chain Monte Carlo
  settings.  We build on this to provide a fast algorithm for
  constructing hyperrectangular confidence regions having the desired
  simultaneous coverage probability and a convenient marginal
  interpretation. The methods are incorporated into standard ways
  of visualizing the results of Monte Carlo experiments enabling the
  practitioner to more easily assess the reliability of the
  results. We demonstrate the utility of this approach in various
  Monte Carlo settings including simulation studies based on
  independent and identically distributed samples and Bayesian
  analyses using Markov chain Monte Carlo sampling.

\smallskip
\noindent \textbf{Keywords.} Asymptotic normality, Monte Carlo, Markov chain Monte Carlo, quantile limit theorems, simulation studies, strong mixing, visualizations.
\end{abstract}

\section{Introduction} \label{sec:intro}

Monte Carlo experiments are a standard tool to estimate features of a
given distribution \citep[see e.g.][]{broo:gelm:jone:meng:2010,
  fish:1996}.  The justification for using Monte Carlo methods is
asymptotic. For example, using a sample mean to estimate a mean of the
distribution is justified by an appropriate strong law.  However, due
to variability arising from simulation in finite time, it is essential
to address whether the simulation has been run sufficiently long so
that the resulting estimates are reliable.

Reporting of results from simulation studies often focuses on
estimates from empirical marginal distributions, e.g., by using sample
boxplots of estimates of one feature of interest.  Popular software
such as \texttt{Stan} \citep{R:rstan} and \texttt{tidybayes}
\citep{R:tidybayes} follow similar procedures for reporting credible
intervals or prediction intervals.  However, consideration of simulation
reliability is rare; see \cite{fleg:hara:jone:2008},
\cite{jone:hara:caff:neat:2006}, and \cite{koeh:brow:hane:2009} for
more discussion.  Despite these works from more than a decade ago, the
situation has not substantially improved.

When simulation reliability is addressed, it usually amounts to either
considering the Monte Carlo sample size (or effective sample size) or,
in the best case, univariate Monte Carlo standard errors.  In Markov
chain Monte Carlo (MCMC) settings, so-called convergence diagnostics
are often reported.  Neither the Monte Carlo sample size (or effective
sample size) nor convergence diagnostics directly assess the
reliability of the resulting estimation.  Univariate Monte Carlo
standard errors ignore the fact that most simulation experiments are
aimed at estimating multiple quantities with the same simulated data,
forcing the use of conservative multiplicity adjustments such as
Bonferroni or Scheff{\' e}, leading to inefficient simulation
practices.

Although there has been recent work on assessing reliability through Monte Carlo error in simultaneous estimation of several expectations in MCMC experiments \citep[see e.g.][]{dai:jones:2017, vats:fleg:jone:2018, vats:fleg:jone:2019}, current best practice does not apply to simultaneous estimation of several quantiles or combinations of means and quantiles, a common goal (e.g.\ in Bayesian applications). Moreover, current multivariate approaches focus on either a relative error interpretation (comparing the Monte Carlo variation to the target distribution variation) or on minimum volume ellipsoidal confidence regions.  The relative error approach applies well in MCMC experiments, but does not translate to standard Monte Carlo, that is, independent and identically distributed (IID) sampling.  The use of ellipsoidal regions makes visualization difficult and complicates the interpretation of marginal estimates.

We aim to address the reliability of simultaneous Monte Carlo estimation of any finite combination of means and quantiles under the assumption that the simulated process is strongly mixing, which includes standard Monte Carlo and correlated sampling from MCMC.  Our approach will also yield a convenient marginal interpretation that, as we will demonstrate, can be easily incorporated into standard existing visualizations of Monte Carlo estimates, making the assessment of estimation reliability convenient to the practitioner.

The techniques we consider provide simultaneous confidence intervals
for any combination of quantiles and expectations with the desired
level, say $1 - \alpha$. Thus we avoid the use of conservative methods
where the overall coverage probability often greatly exceeds
$1-\alpha$.  Our approach begins by establishing asymptotic
(multivariate) normality of any finite collection of sample means and
quantiles.  The covariance matrix of the asymptotic distribution has a
complicated form, but we provide consistent estimators for it. Instead
of proceeding down the well-trodden route of using conservative
hyperrectangular approaches or minimum volume ellipsoidal regions, we
develop a fast algorithm for constructing hyperrectangular regions
having the desired coverage, $1-\alpha$ \citep[cf.][]{wasserman:2010,
  montiel2019simultaneous}.  Using a hyperrectangular approach makes
the marginal interpretation of the resulting intervals
straightforward, making it easier to incorporate in the standard
plots given by existing software.  The following example illustrates
some of these properties.

\begin{example}
  Consider a motivating example with the goal of estimating the mean
  and $(.10,.90)$-quantiles for a three component mixture of normal
  densities. We use a Metropolis-Hastings (MH) random walk to produce
  1000 MCMC samples.  The sample statistics of these draws are
  used to estimate the 3-dimensional quantity of interest and 
  corresponding $3 \times 3$ asymptotic covariance
  matrix. Figure~\ref{fig:mix.intro} shows $90\%$ simultaneous
  confidence intervals superimposed on a plot containing an empirical
  density estimate; notice that the confidence regions displayed
  around each quantity of interest have different lengths enabling
  practitioners to visualize simultaneous simulation variability and
  assess the reliability of the simulation experiment.  It clearly
  illustrates both the variability of the target distribution and the
  variability in Monte Carlo estimation without overemphasizing point
  estimates.  This concludes our initial discussion of this example.

\begin{figure}[tbh]
\begin{center}
\includegraphics[]{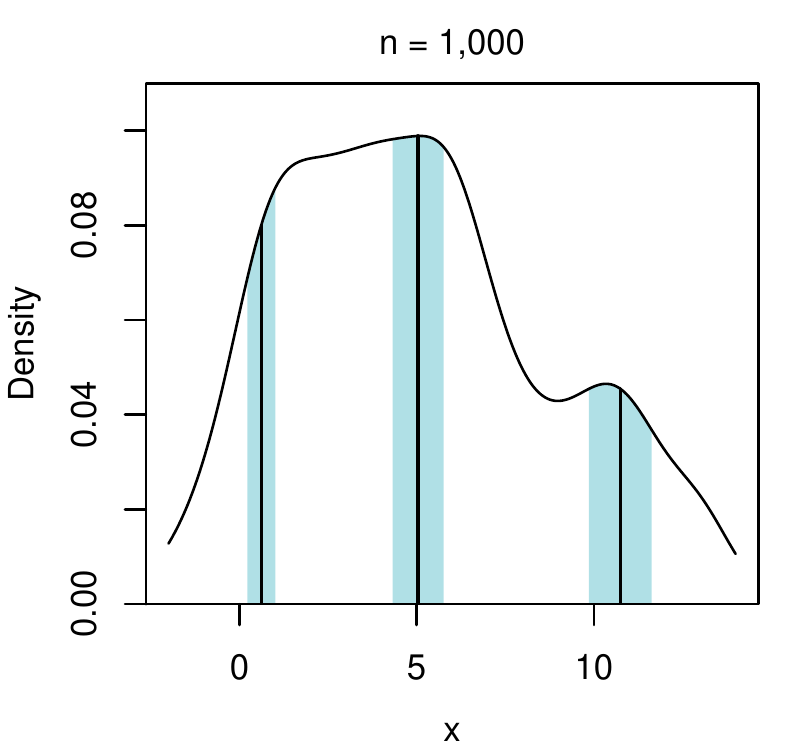}
\caption{Simultaneous confidence intervals for the mean (middle vertical band) and
  $(.10,.90)$-quantiles (outer two vertical bands) based on a MCMC sample
  size of 1000. Notice the substantial simulation variability and that
  the intervals are of different widths.}
\label{fig:mix.intro}
\end{center}
\end{figure}

\begin{comment}
\begin{figure}[tbh]
\begin{center}
  \begin{subfigure}[]{0.49\linewidth}
    \includegraphics[width=\linewidth]{plots/mixnorm_intro_1k}
%    \caption{$n$ = 1,000}
  \end{subfigure}
  \begin{subfigure}[]{0.49\linewidth}
    \includegraphics[width=\linewidth]{plots/mixnorm_intro_50k}
    \caption{$n$ = 50,000}
  \end{subfigure}
\caption{Visualization based on two sample sizes of simultaneous uncertainty bounds for the mean (middle vertical band) and $(.10,.90)$-quantiles (outer two vertical bands).}
\label{fig:mix.intro}
\end{center}
\end{figure}
\end{comment}
\end{example}

The remainder is organized as follows. In Section~\ref{sec:joint_clt},
we establish the joint asymptotic distribution of sample means and
quantiles for strongly mixing processes. Section~\ref{sec:simul_CI}
develops a univariate optimization method for obtaining the
simultaneous confidence intervals described above.  Several examples
are given in Section~\ref{sec:examples}, which demonstrate the
versatility of the approach given here and its potential application
in existing software.  More specifically, we return to the three
component mixture of normals considered above, then turn our attention
to two Monte Carlo simulation studies, and finally consider a Bayesian
analysis example.  In one of the simulation studies we investigate the
empirical coverage properties of the methods we propose.  Some final
remarks are given in Section~\ref{sec:disc}.

\section{Joint asymptotic distribution}
\label{sec:joint_clt}

Suppose $\pi$ is a probability distribution with support $\mathsf{X} \subseteq \mathbb{R}^d, d \geq 1$. Let $p_1$ and $p_2$ be nonnegative integers such that $p_1 + p_2 \ge 1$.  Let $g: \mathsf{X} \to \mathbb{R}^{p_1}$ such that $g=(g_1,\ldots, g_{p_{1}})'$ and, if $X\sim \pi$, then set
\begin{equation*}
\theta_{g_{i}}= E_\pi \left( g_{i}(X) \right) = \int_{\mathsf{X}} g_{i}(x)\pi(dx), 
\end{equation*}
which is assumed to exist.    Collecting the expectations, set
\[
\theta_{g}=(\theta_{g_{1}},\ldots, \theta_{g_{p_{1}}})'.
\]
Let $h:\mathsf{X} \to \mathbb{R}^{p_2}$ such that $h =  \left( h_1, \ldots, h_{p_2} \right)'$. If $V = h_{i}(X)$ with distribution function $F_{h_i}(v)$, which we assume is absolutely continuous with a continuous density $f_{h_i}(v)$, then define the $q_i$-quantile associated with $F_{h_i}$ as
\begin{equation*}
\xi_{q_i} = F_{h_i}^{-1}(q_i) = \inf\{ v:F_{h_i}(v)\geq q_i\}.
\end{equation*}
Collecting the quantiles, set 
\[
\phi_{h} = \left( \xi_{q_1} , \dots , \xi_{q_{p_2}} \right)'.
\]
Then the $(p_{1}+p_{2})$-vector $\nu = (\theta_{g}', \phi_{h}')'$ is the set of features of $\pi$ that we aim to estimate.  

Let $\{X_{n}\}$ be a strictly stationary stochastic process on a
probability space $(\Omega, {\mathcal F}, P)$ and set
${\mathcal F}_{k}^{l} = \sigma(X_{k}, \ldots, X_{l})$, the $\sigma$-algebra generated by $X_{k}, \ldots, X_{l}$.  Define the
$\alpha$-mixing coefficients for $n=1,2,3,\ldots$ as
\[ 
\alpha(n) = \sup_{k \ge 1} \sup_{A \in {\mathcal F}_{1}^{k}, \, B \in {\mathcal F}_{k+n}^{\infty}} | P(A \cap B) - P(A) P(B) | \; .  
\]
Then $\{X_{n}\}$ is \textit{strongly mixing} if $\alpha(n) \to 0$ as
$n \to \infty$ and this will be assumed for the remainder.
Essentially, strong mixing coefficients quantify the rate at which
events in the distant future become independent of the past. Both IID
processes and positive Harris recurrent Markov chains \citep[for
definitions see][]{meyn:twee:1993} are strongly mixing; see
\cite{brad:1986, brad:2005} and \cite{ibra:linn:1971}.

Estimation of $\theta_g$  is straightforward with a vector of sample means  since there is a version of the strong law \citep[see, e.g.,][]{blum:hans:1960, ibra:linn:1971} that ensures, as $n \to \infty$, with probability 1,
\begin{equation*}
\label{eq:SLLN}
\bar{g}_n =\frac{1}{n} \sum_{j=1}^n g(X_j) \to \theta_g .
\end{equation*}
Similarly, estimation of $\phi_{h}$ is straightforward using sample order statistics.  That is, let $\hat{\xi}_{q_i} = h_{i}(X)_{\lceil nq_{i}\rceil:n}$ be the $\lceil nq_{i}\rceil^{th}$ order statistic of $h_{i}(X)$ and denote the vector of $p_2$ estimated quantiles as
\begin{equation*}
\hat{\phi}_n = \left( \hat{\xi}_{q_1} , \dots , \hat{\xi}_{q_{p_2}} \right)'.
\end{equation*}
Define empirical distributions for $F_{h_i}$ as
\[
\bar{F}_{h_i} (v) =\frac{1}{n}\sum_{j=1}^n I(h_{i}(X_j)\leq v),
\]
and
\[
\bar{F}_h ( v ) = \left( \bar{F}_{h_1} (v_1) , \dots, \bar{F}_{h_{p_2}} (v_{p_2}) \right)' .
\]
Letting $Q=(q_{1},\ldots, q_{p_{2}})'$ and noting that $\bar{F}_{h_i} (v)  $ is a sample mean we have that, as $n \to \infty$, $\bar{F}_h ( \phi_h ) \to Q$ with probability 1. This is sufficient to ensure that, as $n \to \infty$, $\hat{\phi}_n \to \phi_{h}$ with probability 1 \citep[for more discussion on this point see][]{doss:fleg:jone:neat:2014}.

While the above asymptotics justify $\hat{\nu}_{n} := (\bar{g}_{n}', \hat{\phi}_{n}')'$  as a simulation-based estimator of $\nu$, no matter how large $n$, there will be an unknown Monte Carlo error $(\hat{\nu}_{n} - \nu)$.
% \[
% \begin{pmatrix} \bar{g}_n - \theta_g \\ \hat{\phi}_n - \phi_h \end{pmatrix} .
% \]
% Fortunately, we can derive the approximate sampling distribution of the Monte Carlo error.
In the following theorem, we present the conditions for an asymptotic sampling distribution of the Monte Carlo error and the explicit form of the covariance matrix.
\begin{theorem} 
\label{thm:joint}
Let $\{X_n\}$ be a stationary $\alpha$-mixing process. Suppose
$F_{h_i}(v)$ is absolutely continuous and twice-differentiable with
density $f_{h_i}(v)$ such that $0 < f_{h_i} (\xi_{q_i}) < \infty$ and
the first derivative $f'_{h_i}(v)$ is bounded in some neighborhood of
$\xi_{q_i}$. In addition, suppose either
\begin{enumerate}[(a)]
	\item there exist $B$ such that $|h_i(x)| < B$ for all $i$ and
          $\alpha(n) = O(n^{-5/2 - \epsilon})$ for some $\epsilon > 0$\, or
	\item $\text{E}_{\pi}\|g\|^{2 + \delta} < \infty$ for some $\delta > 0$ and $\alpha(n) = O(n^{-\kappa})$ for $\kappa > \max\{3, (2 + \delta)/\delta \}$\,.
\end{enumerate}
Let $A_h$ be a $p_2 \times p_2$ diagonal matrix with $i^{th}$ diagonal elements $f_{h_i} (\xi_{q_i})$ and 
\begin{equation*}
\Lambda = \begin{pmatrix}
I_{p_1} & 0_{p_1 \times p_2} \\
0_{p_2 \times p_1} & A_h
\end{pmatrix} .
\end{equation*}
If $Y_{j} = \left( g(X_j) , I(h(X_j) > \phi_h) \right)'$ and
\begin{equation} \label{eq:mcmc.cov}
\Sigma = \Cov (Y_1,Y_1) + \sum_{j= 2}^\infty \left[ \Cov (Y_1,Y_{j}) +
  \Cov (Y_1,Y_{j})' \right] 
\end{equation}
is positive definite, then, as $n \to \infty$,
\begin{equation} \label{eq:CLT:final}
\sqrt{n} \begin{pmatrix} \bar{g}_n - \theta_g \\ \hat{\phi}_n - \phi \end{pmatrix} \overset{d}\to N \left( 0, \Lambda^{-1}\Sigma\Lambda^{-1} \right).
\end{equation}
\end{theorem}

\begin{proof} 
Under condition (a) with $\alpha(n) = O(n^{-5/2 - \epsilon})$, $\sum_n \alpha(n) < \infty$ and under condition (b) with $\alpha = O(n^{-\kappa})$,  $\sum_n \alpha(n)^{\delta/(2 + \delta)} < \infty$ for $\delta > 0$. Define
\[
\Sigma_g = \Cov(g(X_1), g(X_1)) +  \sum_{j=2}^{\infty} \left[\Cov(g(X_1), g(X_j))  + \Cov(g(X_1), g(X_j))' \right]\,.
\]
Then $\Sigma_g$ is the top left $p_1 \times p_1$  principal sub-matrix of $\Sigma$, and thus is positive definite. By \cite{ibra:1962} and \cite{jone:2004}, as $n \to \infty$, 
\begin{equation*}
\label{eq:CLT} 
\sqrt{n}(\bar{g}_n - \theta_g) \overset{d} \to N_p(0, \Sigma_g)\,.
\end{equation*}
Similarly,
\begin{align*}
\Sigma_h &= \Var\left(I(h(X_1) > \phi_h)\right)  \\ 
& \quad \quad +  \sum_{j=2}^{\infty} \left[\Cov(I(h(X_1) > \phi_h), I(h(X_j) > \phi_h))  + \Cov(I(h(X_1) > \phi_h), I(h(X_j) > \phi_h))' \right]
\end{align*}
is positive definite. Thus, due to a joint central limit theorem, there exists a $p_1 \times p_2$ cross-covariance matrices $\Sigma_{gh} = \Sigma_{hg}'$ so that
\begin{equation} 
\label{eq:CLT.joint}
\sqrt{n} \left( 
\begin{pmatrix} \bar{g}_n \\ 1-\bar{F}_h ( \phi_h ) \end{pmatrix} - 
\begin{pmatrix} \theta_g \\ 1-Q \end{pmatrix}
\right) \overset{d}\to N_{p_1+p_2} \left(0, \Sigma =
\begin{pmatrix}
\Sigma_g & \Sigma_{gh} \\
\Sigma_{hg} & \Sigma_h
\end{pmatrix}
\right) .
\end{equation}
Under condition (a) \citep{yosh:1995} or condition (b)
\citep{zha:yang:hu:2014} there is  a Bahadur quantile representation, i.e.
\begin{equation*}
\label{eq:bahadur}
\hat{\xi}_{q_i} = \xi_{q_i} + \frac{\left( 1-\bar{F}_{h_i} (\xi_{q_i}) \right) - (1-q_i)}{f_{h_i} (\xi_{q_i})} + r_{n,q_i}, 
\end{equation*}
where $r_{n,q_i}$ is $o_p(n^{-1/2})$.  Letting $R_n = \left( r_{n,q_1} , \dots, r_{n,q_{p_2}} \right)'$, we have
\begin{equation}\label{eq:bahadur.vec}
\left( 1-\bar{F}_h(\phi_h) \right) - \left( 1 - Q \right) = A_h \left( \hat{\phi}_n - \phi_h \right) + A_h R_n \overset{p}\to A_h \left( \hat{\phi}_n - \phi_h \right).
\end{equation}
Combining \eqref{eq:CLT.joint} and \eqref{eq:bahadur.vec} yields
\begin{equation*}
\begin{split}
\sqrt{n} \left( 
\begin{pmatrix} \bar{g}_n \\ 1-\bar{F}_h ( \phi_h ) \end{pmatrix} - 
\begin{pmatrix} \theta_g \\ 1-Q \end{pmatrix}
\right) 
& = \sqrt{n} \begin{pmatrix} \bar{g}_n - \theta_g \\ A_h \left( \hat{\phi}_n - \phi_h \right) \end{pmatrix} + o_p(1) \\
& = \sqrt{n} \Lambda \begin{pmatrix} \bar{g}_n - \theta_g \\ \hat{\phi}_n - \phi_h \end{pmatrix} + o_p(1) .
\end{split}
\end{equation*}
\end{proof}

\begin{remark}
% The moment condition can be relaxed to $E_\pi |g|^{2+\delta}  < \infty$ where $\delta > 0$ if the mixing coefficients are $\alpha(n) = O(n^{-(5/2)-\epsilon})$ where $\epsilon > 0$ \citep{yosh:1995}, which is a stronger mixing condition. 
\cite{wang2011bahadur} weaken the mixing conditions further, but their Bahadur quantile representation is not applicable for Metropolis-Hastings Markov chains and hence is of limited applicability in MCMC settings.
\end{remark}

Theorem~\ref{thm:joint} can be stated more conveniently for Markov chains.  Suppose $\{X_{n}\}$ is a positive Harris recurrent Markov chain with Markov transition kernel
\[
K^n(x, A) = \Pr (X_{n+j} \in A \mid X_{j}=x) 
\]
that gives the dynamics of the chain. Let $\|\cdot\|$ denote the total variation norm.  Further, let $M: \sX \mapsto \R^+$ with $E_{\pi} M < \infty$, and $\psi: \N \mapsto \R^+$ be
decreasing such that
\begin{equation} \label{eq:tvn}
\Vert K^n (x, \cdot) - \pi(\cdot) \Vert \leq M(x) \psi(n) \; .
\end{equation}
\textit{Polynomial ergodicity of order $m$} where $m > 0$ means \eqref{eq:tvn} holds with $\psi(n) = n^{-m}$.  \textit{Geometric
  ergodicity} means \eqref{eq:tvn} holds with $\psi(n) = t^n$ for some
$0 < t < 1$, and is thus stronger than polynomial ergodicity.   The standard techniques \citep{jone:hobe:2004} for establishing geometric ergodicity ensure $\text{E}_{\pi}M < \infty$, but this is more complicated in the polynomially ergodic setting; see \citet{jone:2004} for some discussion.  It is standard \citep{chan:geye:1994, doss:fleg:jone:neat:2014, jone:2004} that 
\[
\alpha(n) \le  \psi(n) \text{E}_{\pi} M 
\]
which yields the following result.
\begin{corollary}
\label{cor:mcmc}
Suppose $F_{h_i}(v)$ is absolutely continuous and twice-differentiable with density $f_{h_i}(v)$ such that $0 < f_{h_i} (\xi_{q_i}) < \infty$ and $f'_{h_i}(v)$ is bounded in some neighborhood of $\xi_{q_i}$. Let $\{ X_n \}$ be a Harris ergodic Markov chain with stationary distribution $\pi$.  If 
\begin{enumerate}[(a)]
\item there exist $B$ such that $|h_i(x)| < B$ for all $i$ and  $\{ X_n \}$ is polynomial ergodic of order $5/2 + \epsilon$ for some $\epsilon > 0$, or
\item $\{X_n\}$ is polynomially ergodic of order $\kappa  > \max\{3, (2 + \delta)/\delta \}$ where $\delta > 0$ is such that $\text{E}_{\pi}\|g\|^{2 + \delta} < \infty$, or
\item there exist $B$ such that $|h_i(x)| < B$ for all $i$ and  $\{ X_n \}$ is geometrically ergodic, or
\item $\{ X_n \}$ is geometrically ergodic and $E_\pi \|g\|^{2+\delta}  < \infty$ where $\delta > 0$,
\end{enumerate}
then, as $n \to \infty$, \eqref{eq:CLT:final} holds for any initial distribution of the Markov chain. 
\end{corollary}

The conclusion that the claim holds for any initial distribution, even
point masses, holds by an argument similar to that of Proposition
17.1.6 of \cite{meyn:twee:1993}, which says that for positive Harris
Markov chains, if a CLT holds for one initial distribution, then it
holds for every initial distribution.  Next, there has been
significant work on establishing that Markov chains encountered in
statistics are at least polynomially ergodic; see
\cite{doss:hobe:2010}, \cite{ekva:jone:2019}, \cite{hobe:geye:1998},
\cite{hobe:jone:robe:2005}, \cite{jarn:hans:2000},
\cite{jarn:robe:2002}, \cite{john:jone:neat:2013},
\cite{john:jone:2015}, \cite{jone:robe:rose:2014},
\cite{jone:hobe:2004}, \cite{khare:hobe:2013}, \cite{marc:hobe:2004},
\cite{robe:pols:1994}, \cite{tan:jone:hobe:2013},
\cite{tan:hobe:2009}, and \cite{vats:2017}, among many others.

Consider the case where $\{X_{n}\}$ is an IID process from $\pi$.
While it is immediate that $\{X_{n}\}$ is strongly mixing, we can do
better than Theorem~\ref{thm:joint} since weaker density and moment
conditions are required to establish the Bahadur quantile
representation \citep{ghosh1971new}.  The proof the following result
is omitted since it is similar to the proof of
Theorem~\ref{thm:joint}.

\begin{theorem} \label{thm:IID}
Suppose $\{ X_n \}$ is an IID process from $\pi$ with $E_\pi \|g\|^2  < \infty$.  If $0 < f_{h_i} (\xi_{q_i}) < \infty$ and  $f'_{h_i}$ is bounded in a neighborhood of $\xi_{q_i}$ for all $i$, then, as $n \to \infty$, \eqref{eq:CLT:final} holds with $\Sigma = \Var_{\pi}(Y_1)$.
\end{theorem}

Some special cases of Theorem~\ref{thm:IID} have been studied.  Laplace found the joint asymptotic distribution of the sample mean and the sample median \citep{stigler1973studies}.   \cite{ferg:1998} considers a sample mean and an arbitrary quantile \citep[also see][]{lin1980asymptotic} and \cite{babu1988joint} provides an expression for the covariance between two quantiles.

\subsection{Variance estimation}

Making use of Theorems~\ref{thm:joint} and~\ref{thm:IID} requires
estimation of $\Lambda^{-1} \Sigma \Lambda^{-1}$.  Because $\Lambda$
is diagonal with non-zero diagonals $\Lambda^{-1}$ is readily
available.  The only potential difficulty is that we need to estimate
$f_{h_i} (\xi_{q_i})$, but it is easy to use kernel density estimators
with a Gaussian kernel to estimate $f_{h_i}(\hat{\xi}_{q_i})$, and
hence estimate $\Lambda$.  This approach can perform well even for
dependent sequences \citep{doss:fleg:jone:neat:2014}.

The matrix $\Sigma$ requires more attention since IID and strongly
mixing sequences yield different structures of $\Sigma$.  For IID
sampling, set $S_{j} = ( g(X_j) , I(h(X_j) > \hat{\phi}_n)) '$,
$\Sigma$ in Theorem~\ref{thm:IID} can be consistently estimated with
the sample covariance of $\{S_1, S_2, \dots, S_n\}$. For a stationary
strongly mixing sequence, an expression for $\Sigma$ is given by
\eqref{eq:mcmc.cov} and estimating it has been studied using batch
means \citep{chen1987multivariate, jone:hara:caff:neat:2006,
  vats:fleg:jone:2019}, weighted batch means \citep{liu:fleg:2018},
spectral variance \citep{andr:1991,prie:1981, vats:fleg:jone:2018},
initial sequence \citep{dai:jones:2017}, recursive
\citep{chan2017automatic} and regenerative sampling estimators
\citep{seil:1982, hobe:jone:pres:rose:2002}.

In most simulation settings the amount of simulated data potentially
available (letting $p=p_{1}+p_{2}$, then all of $d$, $n$, and $p$ can
be large) forces consideration of the computational effort required to
estimate $\Sigma$.  Batch means has proven to be the most
computationally efficient, there are conditions comparable to those in
Corollary~\ref{cor:mcmc} for strong consistency, and it often
enjoys good finite-sample properties \citep[for more discussion on
these points see][]{chen1987multivariate, vats:fleg:jone:2019}.  Other
methods such as spectral variance estimators can also work well,
especially when $n$ is not too large.  We will restrict our attention
to batch means estimators, which are now defined.  Let $n = ab$ where
$a$ is the number of batches and $b$ is the batch size. The mean for
the $k^{th}$ batch is $\bar{S}_k(b) = b^{-1}\sum_{t=1}^b S_{kb+t}$ and
the overall mean is $\bar{S} = a^{-1}\sum_{k=1}^a \bar{S}_{k}(b)$.
Then the batch means estimator with batch size $b$ is
\[
\hat{\Sigma} = \frac{b}{a-1}\sum_{k=0}^{a-1}
\left(\bar{S}_k(b)-\bar{S} \right) \left(\bar{S}_k(b)-\bar{S} \right)' .
\]
We will typically use a batch size equal to
$\lfloor \sqrt{n} \rfloor$.   If $p$ is large, then batch means may
exhibit downward bias leading to noticeable undercoverage of the
simultaneous confidence intervals developed below. Other variance
estimators such as weighted batch means or a lugsail window function
\citep{vats:fleg:2018} may be used to induce upward bias, but require
a larger computational effort.

\section{Simultaneous confidence intervals}
\label{sec:simul_CI}

We consider $p$-dimensional simultaneous confidence regions for
$\nu$.  Let
$\hat{\Lambda}^{-1} \hat{\Sigma} \hat{\Lambda}^{-1}$ be a strongly
consistent estimator of $\Lambda^{-1} \Sigma \Lambda^{-1}$. Then
Theorem~\ref{thm:joint} allows for the construction of Wald
(ellipsoidal) confidence regions having the desired asymptotic
coverage, $1-\alpha$.
%%%%%%%%%%%%%%%
\begin{comment}
%%%%%%%%%%%%%%%
{\em GJ:  Do we really need this?}
Let $\chi^2_{1-\alpha,p}$ be
the $1-\alpha$ quantile of a $\chi^2_p$ distribution. Then the Wald
confidence region is the following ellipsoidal set
\[
  \left\{\nu: n(\hat{\nu}_{n} - \nu)^T \left(\hat{\Lambda}^{-1}
      \hat{\Sigma} \hat{\Lambda}^{-1}\right)^{-1}(\hat{\nu}_{n} - \nu)
    \leq \chi^2_{1-\alpha, p} \right\}\,.
\]
The above ellipsoid has the right asymptotically coverage,
$1- \alpha$. 
%%%%%%%%%%%%%%%
\end{comment}
%%%%%%%%%%%%%%%
However, visualizing a $p$-dimensional Wald ellipsoid is difficult
beyond two dimensions and is difficult to use in assessing reliability
of Monte Carlo experiments. Instead, we consider hyperrectangular
confidence regions using the information from the full covariance
structure $\Lambda^{-1} \Sigma \Lambda^{-1}$, and not just the
diagonals. The proposed technique is similar to the balanced bootstrap
\citep{beran1988balanced, bruder2018balanced, romano2010balanced} and
the adjusted-Wald intervals \citep{lutk:stas:anna:2015} used primarily
in multiple testing, but we avoid the use of computationally-intensive
naive bootstrap methods.

The basic approach is to consider hyperrectangular regions
$C_{LB} \subseteq C_{UB}$ where $C_{LB}$ is so small that it will have
coverage no greater than $1-\alpha$ while $C_{UB}$ is so large that it
will have coverage at least $1-\alpha$.  There is some
hyperrectangular region, say $C_{\alpha}$, between these, i.e.
$C_{LB} \subseteq C_{\alpha} \subseteq C_{UB}$, which will have
coverage of $1-\alpha$.  In general, finding $C_{\alpha}$ will be
computationally intensive, but we provide a simple solution
by considering a particular subclass of regions using the results of
Section~\ref{sec:joint_clt}.

Let $\hat{\sigma}_{g_i}$ be the $i$th diagonal of $\Sigma_g$,
$i = 1, \dots, p_1$ and $\hat{\gamma}_j$ be the $j$th diagonal of
$\hat{\Lambda}^{-1} \hat{\Sigma}_h \hat{\Lambda}^{-1}$,
$j = 1, \dots, p_2$. For $z > 0$, consider the hyperrectangular
confidence regions of the form
\[
  C_{SI}(z) = \prod_{i=1}^{p_1} \left[ \bar{g}_{n,i} - z
    \dfrac{\hat{\sigma}_{g_i} }{n}, \bar{g}_{n,i} + z
    \dfrac{\hat{\sigma}_{g_i} }{n} \right] \prod_{j=1}^{p_2} \left[
    \hat{\phi}_{n,j} - z \dfrac{\hat{\gamma}_{j} }{n},
    \hat{\phi}_{n,j} + z \dfrac{\hat{\gamma}_{j} }{n} \right]\,.
\]
Setting $z = z_{1-\alpha/2}$ yields uncorrected intervals that
simultaneously have coverage no greater than $1-\alpha$ and we set
$C_{LB}: = C_{SI}(z_{1-\alpha/2})$. With $z = z_{1-\alpha/2p}$, we get
the Bonferroni-corrected hyperrectangular region which has coverage at
least $1 - \alpha$ and we set $C_{UB} := C_{SI}
(z_{1-\alpha/2p})$. While there are many other conservative approaches,
 such as the \u{S}id\'{a}k correction \citep{vsidak:1967},
Scheffe's correction \citep{scheffe:1999s}, and Wald projections
\citep{lutk:stas:anna:2015}, we have found the Bonferroni approach to
be an effective way to construct $C_{UB}$.

All that is left is to construct $C_{\alpha}$, but this is easily done
if we can find $z^*$ such that
$z_{1-\alpha} \leq z^* \leq z_{1-\alpha/2p}$ and $C_{SI}(z^*)$ has
coverage $1-\alpha$, in which case we set
$C_{\alpha}=C_{SI}(z^*)$. The approach to finding $z^*$ is illustrated
in Figure~\ref{fig:sim.int}, but more formally, the optimization
algorithm works as follows: suppose
$U \sim N_p(\hat{\nu}_{n}, \hat{\Lambda}^{-1} \hat{\Sigma}
\hat{\Lambda}^{-1}/n)$, then $\Pr(U \in C_{LB}) \le (1-\alpha)$ and
$\Pr(U \in C_{UB}) \ge (1-\alpha)$.  Since
$\Pr(U \in \mathcal{C}_{SI}(z))$ is strictly increasing as $z$
increases, we use the bisection method between $z_{1-\alpha/2}$ and
$z_{1-\alpha/2p}$ to find $z^*$ such that
$\Pr(U \in \mathcal{C}_{SI}(z^*)) \approx (1-\alpha)$ and set
$C_{\alpha}= \mathcal{C}_{SI}(z^*)$.  Thus finding the $p$-dimensional
region $C_{SI}(z^*)$ only requires optimizing over a univariate
parameter, $z$.  The required multivariate normal probabilities can be
quickly and accurately calculated using quasi-Monte Carlo methods
\citep{R:mvtnorm, genz2009computation}.

% If $C_{LB} = \mathcal{C}_{SI}(z_{1-\alpha/2})$, we have $P(U \in C_{LB}) \le (1-\alpha)$ and for $C_{UB} = \mathcal{C}_{SI}(z_{1-\alpha/2p})$, we have $P(U \in C_{UB}) \ge (1-\alpha)$.  Since $P(U \in \mathcal{C}_{SI}(z))$ is strictly increasing as $z$ increases, we use the bisection method between $z_{1-\alpha/2}$ and $z_{1-\alpha/2p}$ to find $z^*$ such that $P(U \in \mathcal{C}_{SI}(z^*)) \approx (1-\alpha)$ and set $C_{\alpha}= \mathcal{C}_{SI}(z^*)$.  The required multivariate normal probabilities can be quickly and accurately calculated using quasi-Monte Carlo methods \citep{R:mvtnorm, genz2009computation}.

%%%%%%%%%%%%
\begin{comment}
%%%%%%%%%%%%

\textit{GJ:  Maybe this would be more appropriate for supplementary material?}
efficient numerical solvers such as the function \texttt{fzero} from the \texttt{R} package \texttt{pracma} \citep{borchers2019pracma} would also be appropriate.

The solution will be up to some error tolerance $\epsilon = P(U \in \mathcal{C}_{SI}(z^*)) - (1-\alpha)$.  The choice of $\epsilon$ deserves some care depending on how $P(U \in \mathcal{C}_{SI}(z))$ is calculated.  To this end, we conducted several simulations varying covariance, dimension, and probability values. The largest error provided by the function \texttt{pmvnorm} was approximately .002 with most errors less than .001. The \texttt{R} code for this study is available as part of the supplementary material.  We recommend setting $\epsilon = .001$ or $\epsilon = .002$.

%%%%%%%%%%%
\end{comment}
%%%%%%%%%%%

\begin{figure}[!htb]
\begin{center}
\includegraphics[width=10cm]{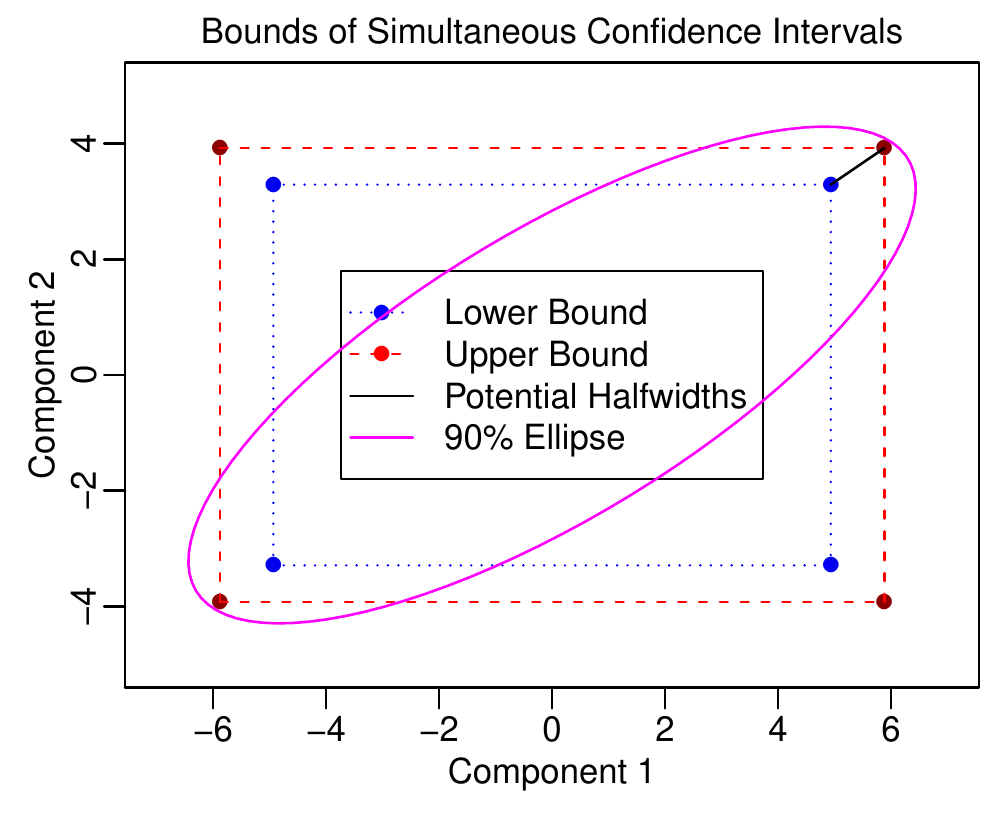}
\caption{Plot of $C_{LB}$ (blue dotted) and $C_{UB}$ (red dashed) from a $90\%$ confidence region for a bivariate normal distribution with component variances 9 and 4.  The black line in the first quadrant indicates the potential search values to achieve the desired coverage level.}
\label{fig:sim.int}
\end{center}
\end{figure}

\section{Applications}\label{sec:examples}

This section considers implementation of the methodology in three
examples.  The first one expands on the mixture model from
Section~\ref{sec:intro} demonstrating use of the methods in both IID
sampling and MCMC settings and investigating the finite-sample
properties of the proposed confidence regions.  The second application
is to a setting where we compare the finite-sample properties of three
frequentist statistical models, and the third demonstrates the use of our
methods in a Bayesian hierarchical model. In each application, we
identify a combination of means and quantiles of interest, obtain
$1- \alpha$ level simultaneous confidence intervals, and integrate the
intervals into a plot which is standard in the respective
applications.  The full \texttt{R} code is available as part of the
supplementary material.

\subsection{Mixture of normals}

For $j=1,2,3$, let $f_j(x;\mu_j,\sigma_j^2)$ be a normal density with
mean $\mu_j$ and variance $\sigma_j^2$. Set
\begin{equation*}\label{eq:mix:3comp}
\pi(x) = .3f_1(x;1,2.5) + .5f_2(x;5,4) + .2f_3(x;11,3).
\end{equation*}
We consider simultaneous estimation of the mean, $.10$ quantile, and
$.90$ quantile denoted $\mu$, $\xi_{.10}$, and $\xi_{.90}$,
respectively, and hence
$\nu = (\mu, \xi_{.10}, \xi_{.90})$.  Notice that we can
use numerical integration methods to find that
$ (\mu, \xi_{.10}, \xi_{.90}) = (5, .2544116, 11.0143117)$
with absolute error less than 2.5e-6.

We consider both standard Monte Carlo (IID samples) and MCMC metods
for simulating from $\pi$.  Producing IID samples is easy, simply
choose a component according to its mixture probability and then
simulate from it.  In the MCMC setting we use a random walk
Metropolis-Hastings sampler with a $N(0, 9)$ proposal distribution.
In either case we estimate $\nu$ and
$\Lambda^{-1} \Sigma \Lambda^{-1}$ as described in
Section~\ref{sec:joint_clt}. We then calculate
$C_{SI}(z^\ast)$ at the 90\% confidence level and estimate the density
to create Figure~\ref{fig:mix} (the lower left panel appeared in the
introduction). The estimates of $\nu$ are represented by
purple lines with the blue region around each estimate representing
the simultaneous confidence regions.  Notice that as the number of
Monte Carlo samples increases, the width of the intervals decreases.
Also, as should be expected, the width of the intervals based on MCMC
samples is wider than those based on IID samples. Finally, note that
because the shape of the density is asymmetric, the two quantiles
occur at different density values. This affects the value of
$\Lambda^{-1}$ and contributes to the different lengths of the error
regions around each estimate.

\begin{figure}[!hbt]
\begin{center}
    \includegraphics[width = 2.6in]{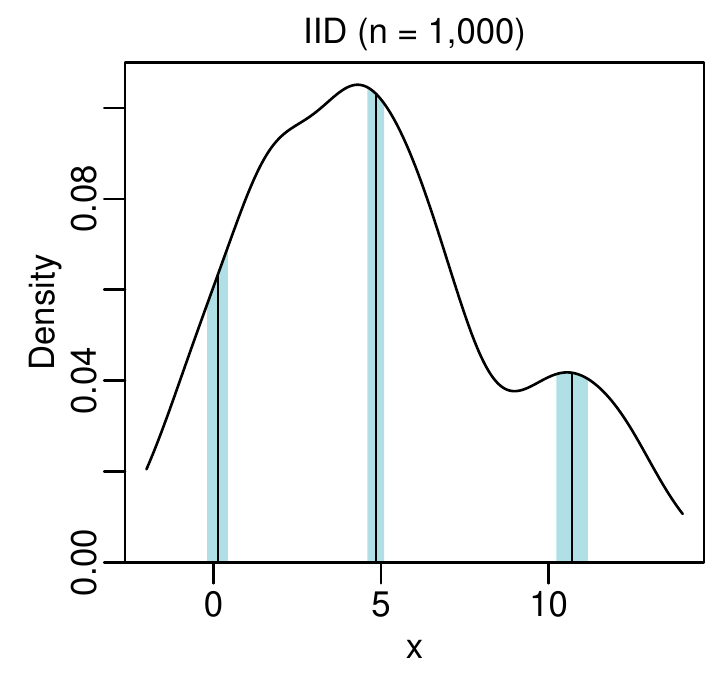}
    \includegraphics[width = 2.6in]{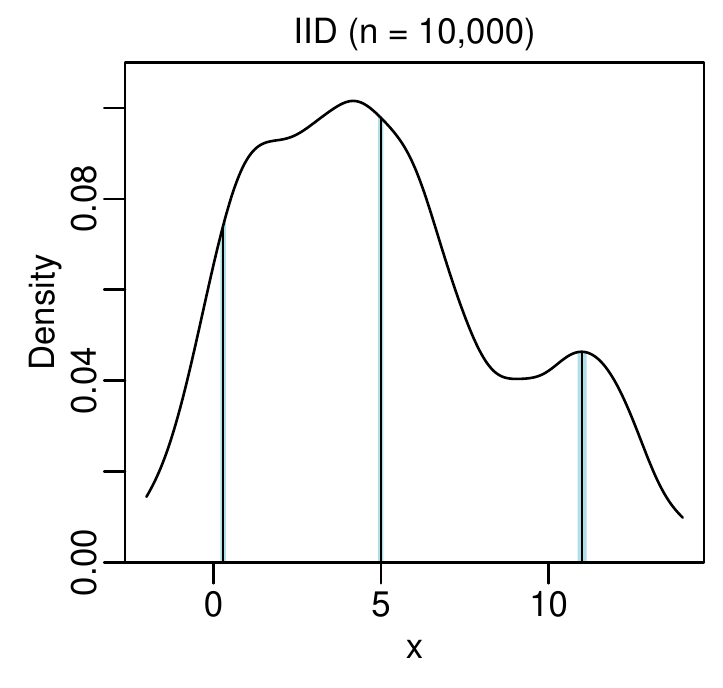}
    \includegraphics[width = 2.6in]{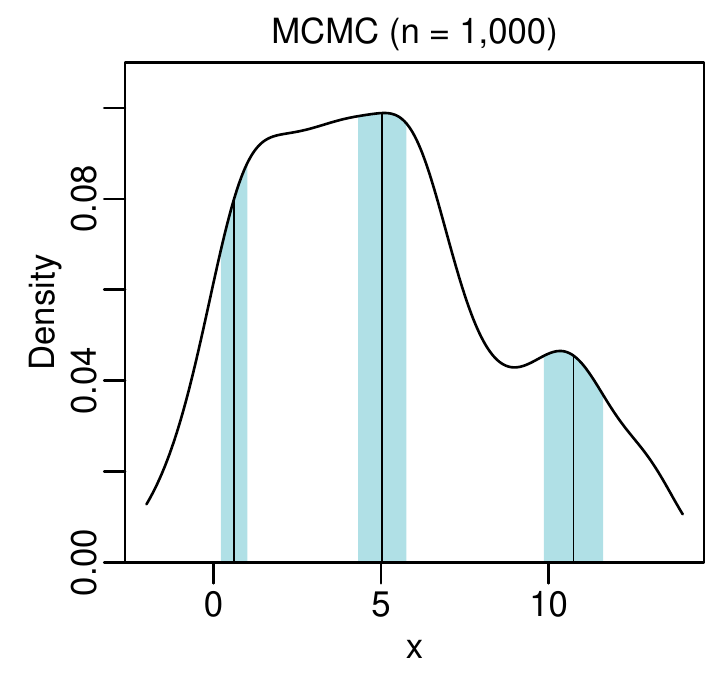}
    \includegraphics[width = 2.6in]{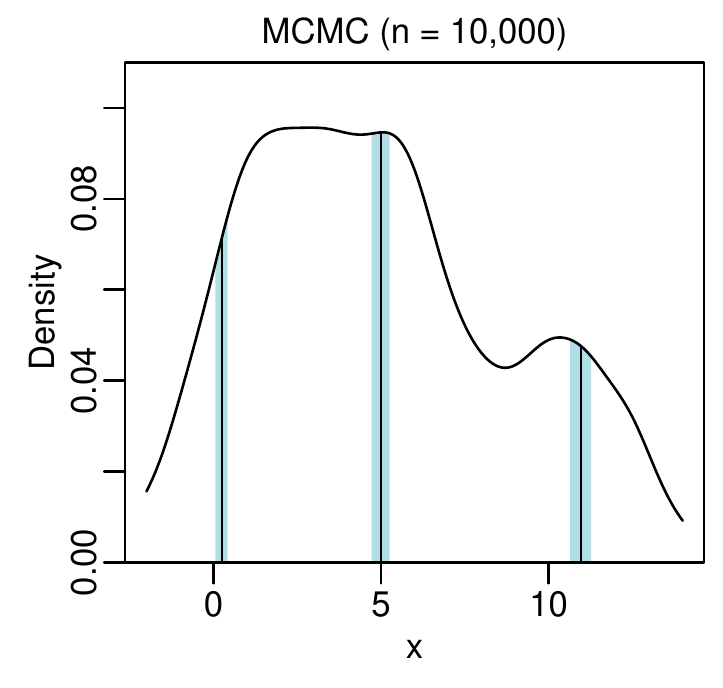}
\caption{Simultaneous 90\% confidence intervals of the mean, $.10$ quantile, and $.90$ quantile from a mixture of normal distributions.}
\label{fig:mix}
\end{center}
\end{figure}

Figure~\ref{fig:mix} concerned a single Monte Carlo simulation where 
every interval presented (12 total) contained the true parameter.  We now turn our
attention to investigating the finite-sample simultaneous coverage
probabilities of the intervals $\mathcal{C}_{SI}(z^*)$.  More
specifically, we do both the IID and MCMC sampling for 2000
independent replications with $n=10000$.  In each replication we calculate
simultaneous confidence intervals $\mathcal{C}_{SI}(z^*)$, uncorrected
marginal intervals $C_{LB}$, and simultaneous Bonferroni intervals
$C_{UB}$ at the 80\% and 90\% confidence levels for which we record
whether each region contains the true value. Thus, we have six binary
outcomes, one for each region and confidence level combination, which
are naturally correlated since we are using the same simulated data.
We then use our procedure, based on Theorem~\ref{thm:joint}, with
overall 95\% confidence level to plot the point estimate of the
coverage along with simultaneous confidence regions in
Figure~\ref{fig:cover}.  %This procedure is repeated for each of $n \in \{1000, 10000\}$.  

\begin{figure}[!htb]
\begin{center}
\includegraphics[]{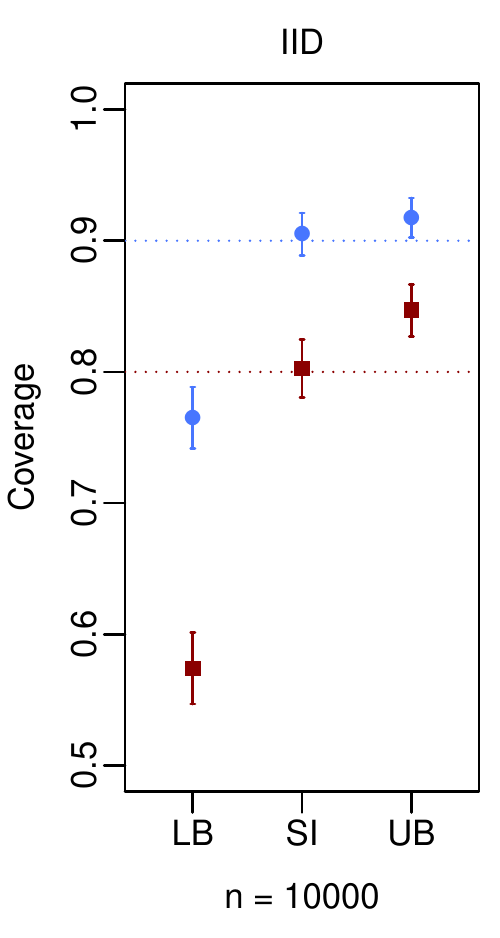}
\includegraphics[]{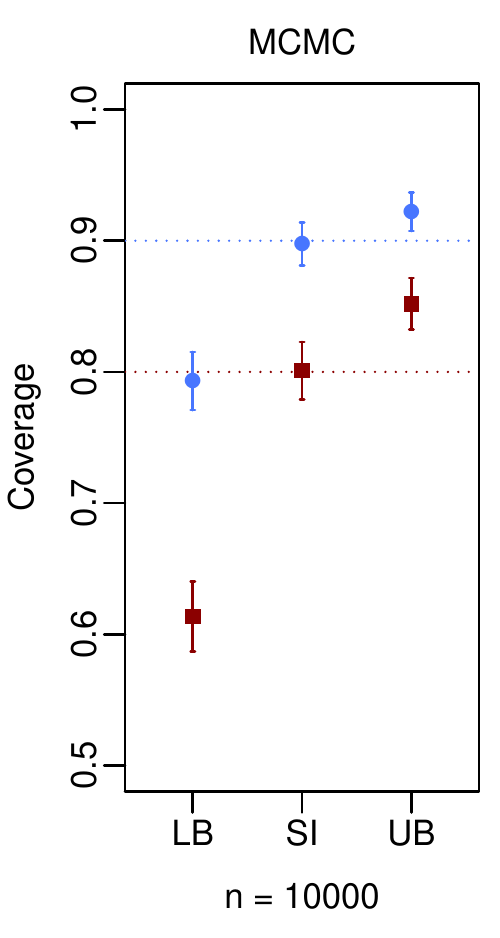}
\caption{Simultaneous 95\% confidence intervals for coverage
  probabilities based on 2,000 replications comparing uncorrected
  marginal intervals $C_{LB}$, simultaneous confidence intervals
  $\mathcal{C}_{SI}(z^*)$, and simultaneous Bonferroni intervals
  $C_{UB}$. Blue intervals with a circle and dark red intervals with a
  square correspond to .9 and .8 nominal levels, respectively.}
\label{fig:cover}
\end{center}
\end{figure}

From Figure~\ref{fig:cover} we can see that $C_{LB}$ yields
significant undercoverage while failing to ever capture the nominal
coverage probability within any of its interval estimates. For
$C_{SI}(z^\ast)$, the confidence intervals contain the nominal level
% as the sample size increases 
illustrating that the simultaneous confidence
intervals yield coverage close to the nominal level.  Bonferroni
intervals, $C_{UB}$, overestimate the nominal level although this
overcoverage is relatively small since the adjustment is based on a
small number of quantities.  However, estimation procedures of higher
dimensionality will correspond to more conservative estimates for the
upper bound.  

\subsection{Comparing methods with boxplots}

Finite-sample performance of statistical methodologies is often done via loss function comparisons with existing methods over repeated
simulations with results presented in side-by-side boxplots. Our
methodology can be used to illustrate the variability in the
estimation of the quantiles in the boxplots and provides a tool to
assess whether sufficient replications have been used or if observed
differences are subject to substantial Monte Carlo error.

\begin{figure}[!htb]
\begin{center}
  \begin{subfigure}[]{0.38\linewidth}
    \includegraphics[width=\linewidth]{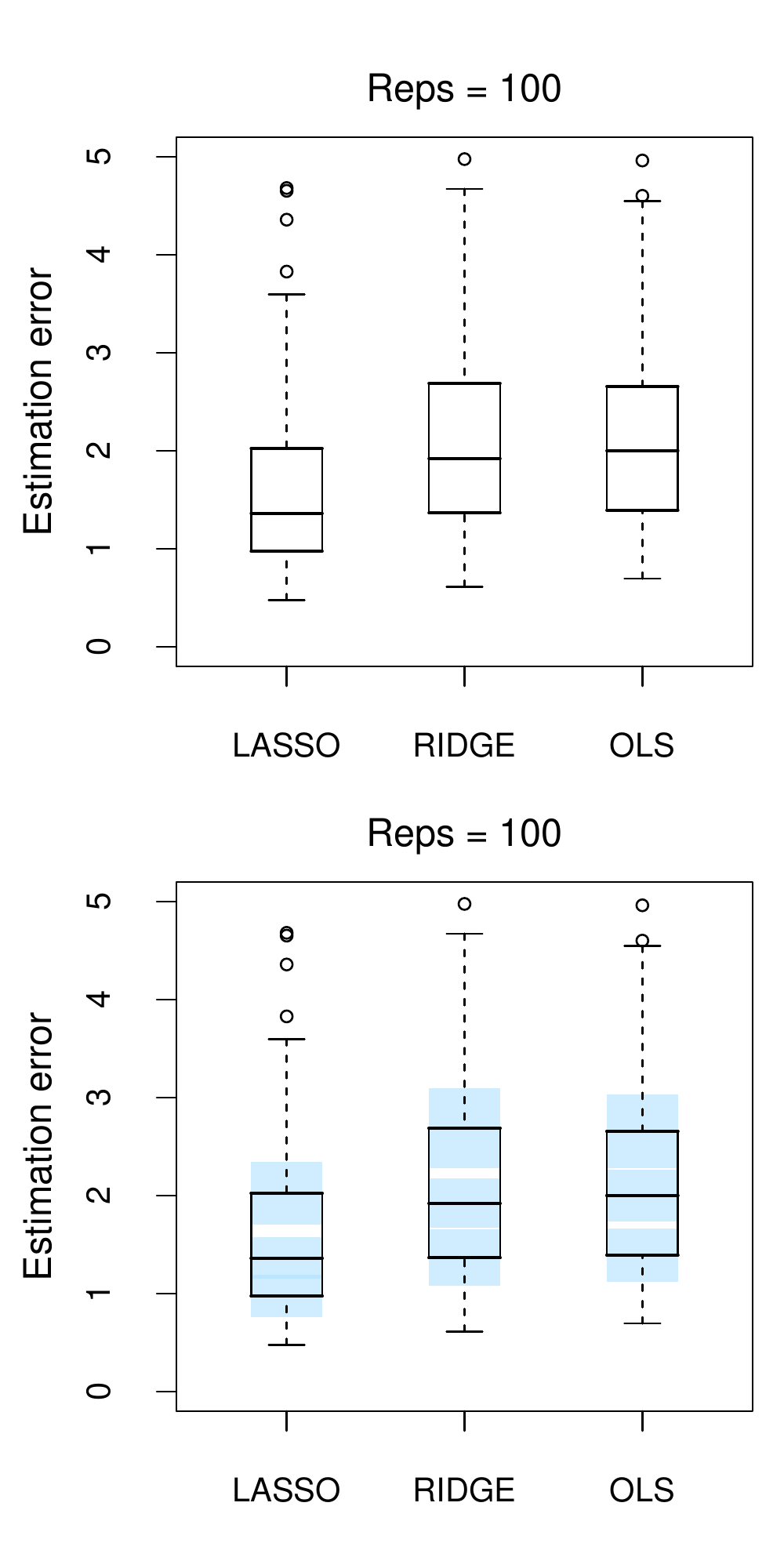}
  \end{subfigure}
  \begin{subfigure}[]{0.38\linewidth}
    \includegraphics[width=\linewidth]{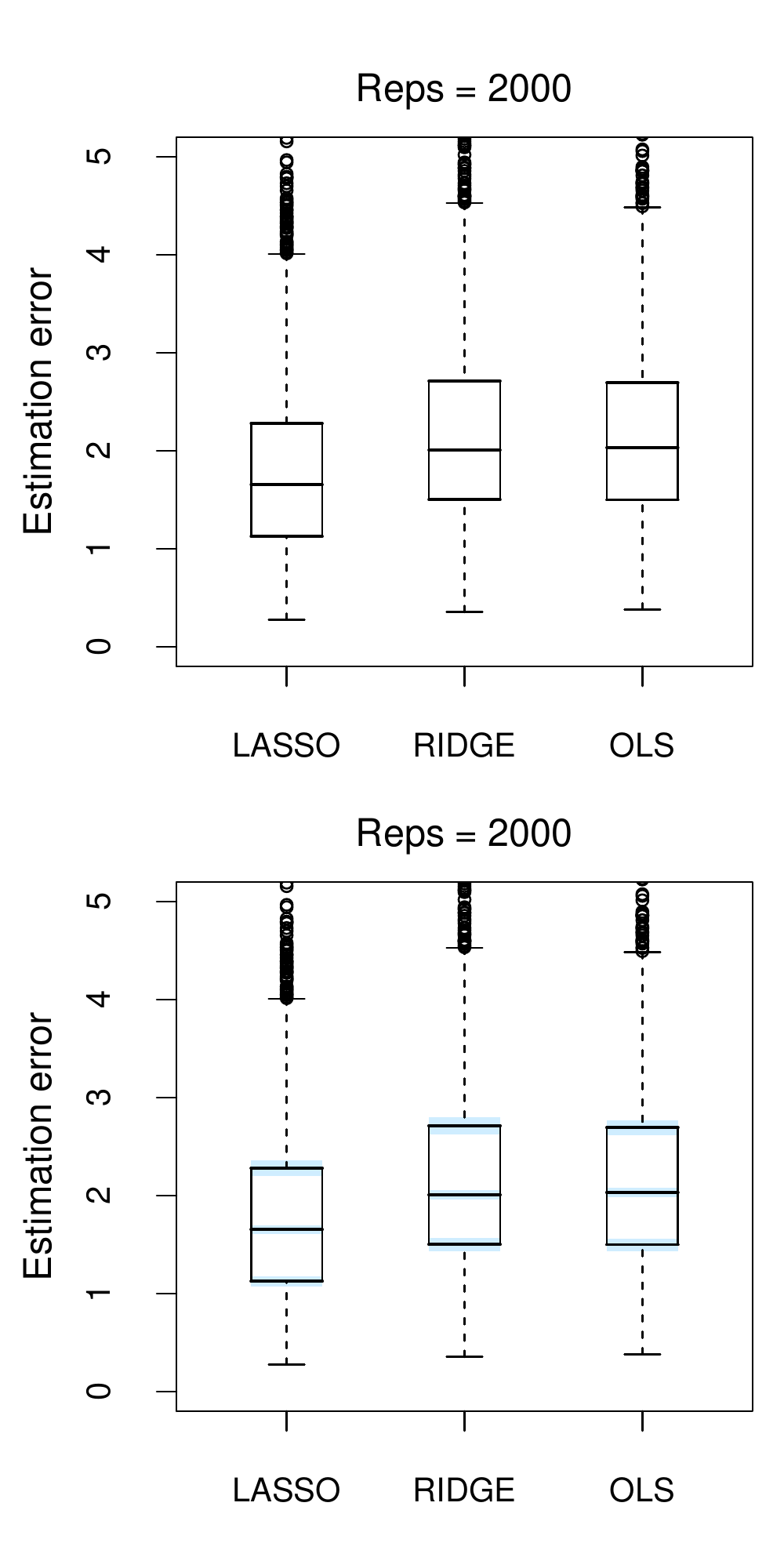}
  \end{subfigure}
\caption{Boxplots of squared estimation error for lasso, ridge, and
  OLS with and without simultaneous confidence intervals.} 
\label{fig:MC_simstudy}
\end{center}
\end{figure}
%%%%%%%%%%%%%%%%%%%
\begin{comment}
%%%%%%%%%%%%%%%%%%%
\begin{figure}[!htb]
\begin{center}
%\includegraphics[width=10cm]{plots/simstudy100.pdf}\\
%\includegraphics[width=10cm]{plots/simstudy500.pdf}\\
%\includegraphics[width=10cm]{plots/simstudy2000.pdf}
  \begin{subfigure}[]{0.32\linewidth}
    \includegraphics[width=\linewidth]{plots/simstudy100.pdf}
  \end{subfigure}
  \begin{subfigure}[]{0.32\linewidth}
    \includegraphics[width=\linewidth]{plots/simstudy500.pdf}
  \end{subfigure}
  \begin{subfigure}[]{0.32\linewidth}
    \includegraphics[width=\linewidth]{plots/simstudy2000.pdf}
  \end{subfigure}
\caption{Boxplots of squared estimation error for lasso, ridge, and OLS with and without simultaneous confidence intervals. Monte Carlo sample size increases from left to right.}
\label{fig:MC_simstudy}
\end{center}
\end{figure}
%%%%%%%%%%%%%%%%%%%
\end{comment}
%%%%%%%%%%%%%%%%%%%
Consider a comparison of lasso \citep{tibshirani1996}, ridge
\citep{hoerl1970ridge}, and ordinary least squares (OLS)
regressions. We will sample $n \in \{100, 2000\}$ independent vectors
of observations $y \in \mathbb{R}^{100}$ from the following model.
Let $Z$ be a $100 \times 21$ dimensional matrix of covariates, and
$\beta^* \in \mathbb{R}^{21}$ be the true regression coefficient
vector. For $\epsilon \sim N_{100}(0, I_{100})$, our data generating
model is
\[
Y = Z\beta^* + \epsilon .
\]
We set $\beta^*$ to be such that the first 11 elements are zero, and
the last 10 are random draws from a normal distribution with mean 0
and variance 2. The matrix $Z$ is constructed such that the first
column is all 1s, and the rows of $Z_{-1}$, the matrix $Z$ with the
first column removed, are drawn from $N_{20}(0, \Omega)$, where the
$ij$th entry of $\Omega$ is $.90^{|i-j|}$. 

For each of the simulated $y \in \mathbb{R}^{100}$ we fit lasso,
ridge, and OLS regressions to estimate the vector of
coefficients. Lasso and ridge estimates are obtained using the
\texttt{glmnet} package \citep{friedman2009glmnet} with tuning
parameters chosen using cross-validation. In each replication, we note
the squared estimation error of the estimated coefficient,
$\hat{\beta}$, that is, $\|\hat{\beta} - \beta^*\|^2$.

Figure~\ref{fig:MC_simstudy} presents boxplots of the squared
estimation error over replications with and without simultaneous
confidence intervals. Recall a box in the boxplot has 25\%, 50\%, and
75\% quantiles. To construct simultaneous confidence intervals, we
appeal to the 9-dimensional joint asymptotic distribution for IID
sequences for these quantiles. The simultaneous confidence intervals
immediately indicate that with only 100 Monte Carlo replications, all
quantile estimates have large variation and that the Monte Carlo error
is swamping any differences between the estimated quantiles of all
three methods. This variability is substantially smaller with 2000
Monte Carlo replications where we can see that the Monte Carlo error
is sufficiently small so that we can claim that ridge and OLS perform
similarly while the lasso is superior. Incorporating our
simultaneous confidence regions allows us to assess the reliability
of the simulations results in a way that is impossible with the top row in Figure~\ref{fig:MC_simstudy}.

\subsection{Bayesian analyses}
\label{sec:bayes}
Consider \pcite{rubi:1981} school data. At eight schools students
were put into SAT prep classes or in a control group where they did not
receive coaching.  The data reported is the estimated effect of prep
classes on verbal SAT scores and the standard deviation at each
school. We are interested in estimating the coaching effect at each
school.

A standard model \citep{gelm:carl:ster:rubi:2004} for
analyzing this data is as follows.  For $j = 1, \dots, J$, let
\begin{equation*}
Y_j \mid \theta_j \stackrel{ind}{\sim} N(\theta_j, \sigma_j^2)
\end{equation*}
\begin{equation*}
\theta_j \stackrel{ind}{\sim} N(\mu, \tau^2)\,,
\end{equation*}
with $\sigma^2_j$ known and priors $f(\mu) \propto 1$ and
$f(\tau) \propto 1/\tau$. We are interested in producing interval
estimates of the coaching effect for each school, that is the
$\theta_{j}$'s. To this end we simulate draws from the joint
posterior $\theta, \mu, \tau|y$ with a deterministic scan Gibbs
sampler described in the supplementary material and estimate posterior
credible intervals for each $\theta_{j}$.

In Figure~\ref{fig:school.large.var} we plot estimated marginal
densities and 80\% credible intervals for each of the $\theta_{j}$'s
based on $n=10,000$ MCMC samples.  We then use our procedure, based
on Theorem~\ref{thm:joint} and Section~\ref{sec:simul_CI}, to include
simultaneous 90\% confidence intervals for the 16 marginal
quantiles. Figure~\ref{fig:school} presents the same analysis based on
an MCMC sample of $n=100,000$.  Comparing Figures~\ref{fig:school.large.var}
and~\ref{fig:school} illustrates that accounting for the Monte Carlo
error can impact our statistical conclusions about the coaching
effects.  Consider the left endpoints of the credible regions in the
plots for $\theta_2$, $\theta_4$, and $\theta_8$.  Specifically, in
Figure~\ref{fig:school.large.var} the left endpoints are
indistinguishable from zero when we account for the Monte Carlo error,
but this is no longer an issue with Figure~\ref{fig:school}.

%%%%%%%%%%%%%%%%%%%
\begin{comment}
%%%%%%%%%%%%%%%%%%%
{\em Let's move this to the supplementary material.}
using the following full
conditional distributions
\begin{gather*}
\theta_j|\mu,\tau,y \sim N\left(\frac{y_j \tau^2 + \mu \sigma_j^2}{\tau^2+\sigma_j^2}, \frac{1}{\frac{1}{\sigma_j^2}+\frac{1}{\tau^2}}\right), \; \;
\mu | \theta,\tau,y \sim N\left(\bar{\theta}, \frac{\tau^2}{J}\right), \text{ and} \\
\tau^2| \theta,\mu,y \sim \text{Inv}-\chi^2\left(J-1,\frac{1}{J-1}\sum_{j=1}^J(\theta_j-\mu)^2\right)\,,
\end{gather*}
where $\bar{\theta}$ is the average of the $\theta_j$'s. 
%%%%%%%%%%%%%%%%%%%
\end{comment}
%%%%%%%%%%%%%%%%%%

\begin{figure}[tbh]
\begin{center}
\includegraphics[width=4in,height=5in]{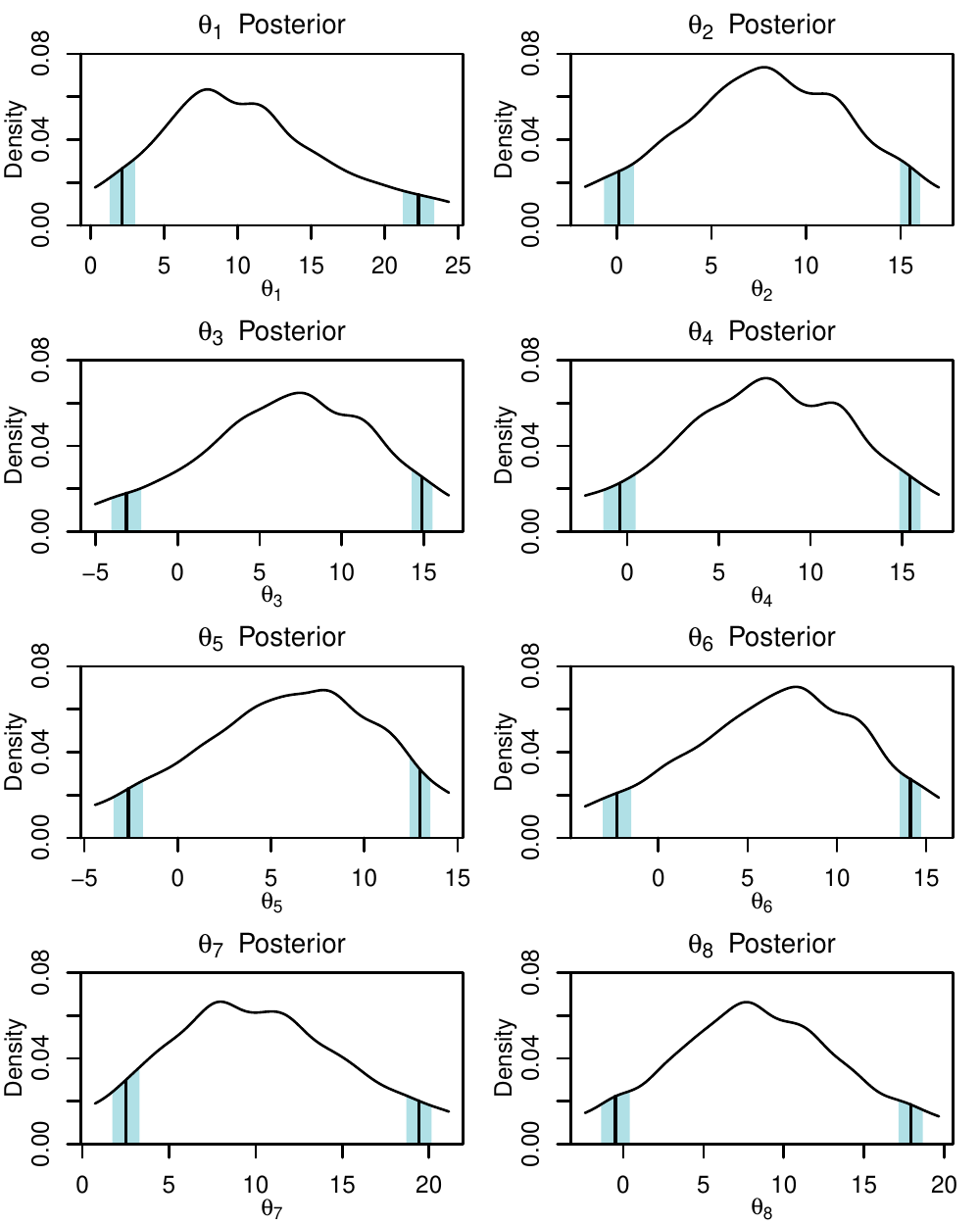}
%\includegraphics[width=\linewidth]{plots/school10k/10k_part1}
%\includegraphics[width=\linewidth]{plots/school10k/10k_part2}
%  \begin{subfigure}[]{0.24\linewidth}
%    \includegraphics[width=\linewidth]{plots/school10k/credTheta1_10k}
%    \caption{$\theta_1$}
%  \end{subfigure}
%  \begin{subfigure}[]{0.24\linewidth}
%    \includegraphics[width=\linewidth]{plots/school10k/credTheta2_10k}
%    \caption{$\theta_2$}
%  \end{subfigure}
%  \begin{subfigure}[]{0.24\linewidth}
%    \includegraphics[width=\linewidth]{plots/school10k/credTheta3_10k}
%    \caption{$\theta_3$}
%  \end{subfigure}
%  \begin{subfigure}[]{0.24\linewidth}
%    \includegraphics[width=\linewidth]{plots/school10k/credTheta4_10k}
%    \caption{$\theta_4$}
%  \end{subfigure}
%  \begin{subfigure}[]{0.24\linewidth}
%    \includegraphics[width=\linewidth]{plots/school10k/credTheta5_10k}
%    \caption{$\theta_5$}
%  \end{subfigure}
%  \begin{subfigure}[]{0.24\linewidth}
%    \includegraphics[width=\linewidth]{plots/school10k/credTheta6_10k}
%    \caption{$\theta_6$}
%  \end{subfigure}
%  \begin{subfigure}[]{0.24\linewidth}
%    \includegraphics[width=\linewidth]{plots/school10k/credTheta7_10k}
%    \caption{$\theta_7$}
%  \end{subfigure}
%  \begin{subfigure}[]{0.24\linewidth}
%    \includegraphics[width=\linewidth]{plots/school10k/credTheta8_10k}
%    \caption{$\theta_8$}
 % \end{subfigure}
\caption{Plot of the estimates of an $80\%$ credible interval for each $\theta$ with simultaneous $90\%$ confidence intervals for 10,000 samples.}
\label{fig:school.large.var}
\end{center}
\end{figure}

\begin{figure}[htb]
\begin{center}
\includegraphics[width=4in,height=5in]{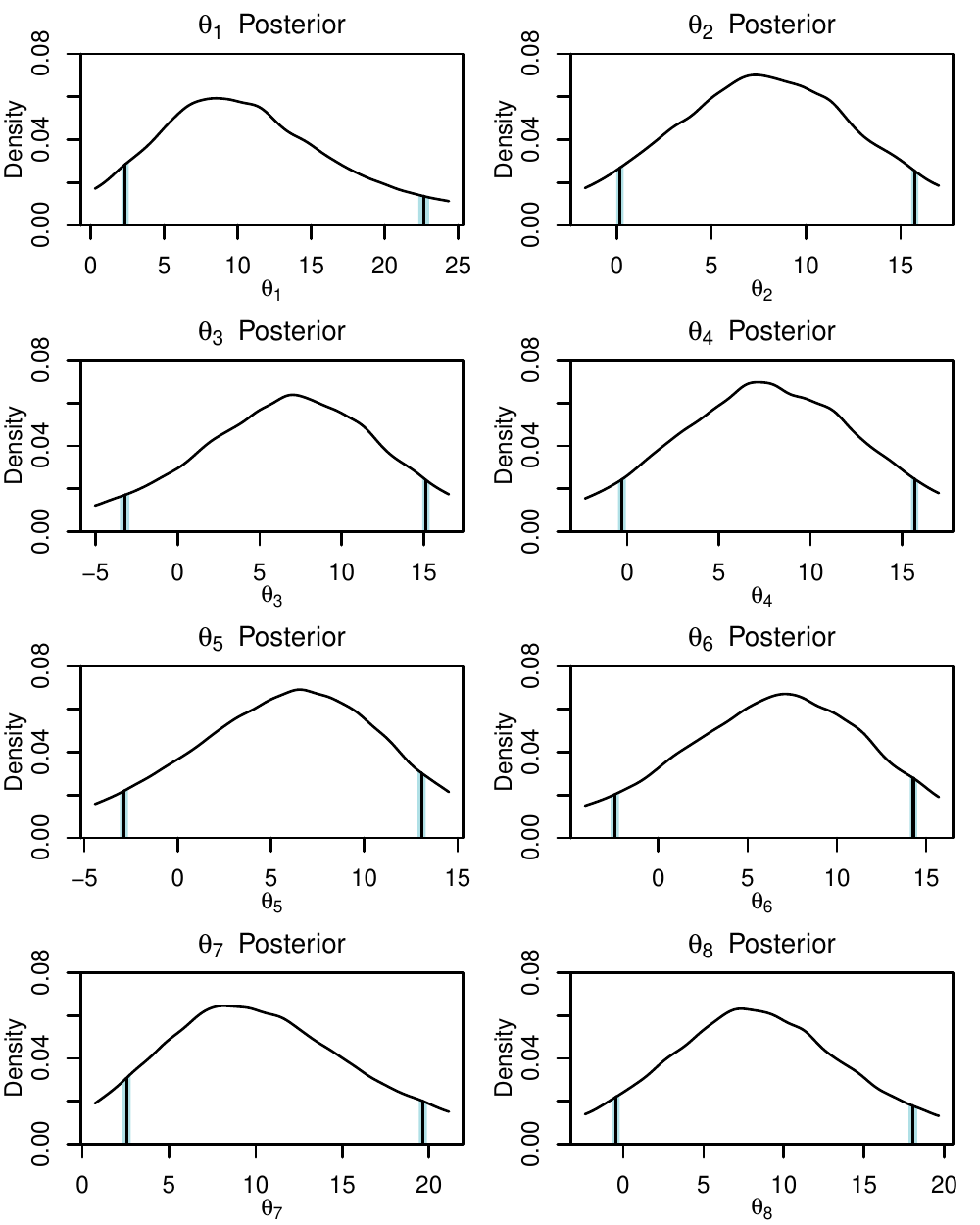}
%\includegraphics[width=\linewidth]{plots/school100k/100k_part1}
%\includegraphics[width=\linewidth]{plots/school100k/100k_part2}
%  \begin{subfigure}[]{0.24\linewidth}
%    \includegraphics[width=\linewidth]{plots/school100k/credTheta1_100k}
%    \caption{$\theta_1$}
%  \end{subfigure}
%  \begin{subfigure}[]{0.24\linewidth}
%    \includegraphics[width=\linewidth]{plots/school100k/credTheta2_100k}
%    \caption{$\theta_2$}
%  \end{subfigure}
%  \begin{subfigure}[]{0.24\linewidth}
%    \includegraphics[width=\linewidth]{plots/school100k/credTheta3_100k}
%    \caption{$\theta_3$}
%  \end{subfigure}
%  \begin{subfigure}[]{0.24\linewidth}
%    \includegraphics[width=\linewidth]{plots/school100k/credTheta4_100k}
%    \caption{$\theta_4$}
%  \end{subfigure}
%  \begin{subfigure}[]{0.24\linewidth}
%    \includegraphics[width=\linewidth]{plots/school100k/credTheta5_100k}
%    \caption{$\theta_5$}
%  \end{subfigure}
%  \begin{subfigure}[]{0.24\linewidth}
%    \includegraphics[width=\linewidth]{plots/school100k/credTheta6_100k}
%    \caption{$\theta_6$}
%  \end{subfigure}
%  \begin{subfigure}[]{0.24\linewidth}
%    \includegraphics[width=\linewidth]{plots/school100k/credTheta7_100k}
%    \caption{$\theta_7$}
%  \end{subfigure}
%  \begin{subfigure}[]{0.24\linewidth}
%    \includegraphics[width=\linewidth]{plots/school100k/credTheta8_100k}
%    \caption{$\theta_8$}
%  \end{subfigure}
\caption{Plot of the estimates of an $80\%$ credible interval for each $\theta$ with simultaneous $90\%$ confidence intervals for 100,000 samples.}
\label{fig:school}
\end{center}
\end{figure}

\section{Final Remarks} \label{sec:disc}

Simultaneous estimation of means and quantiles has received little
attention in the simulation literature, despite being common
practice. We provide an approach for assessing the quality of
estimation for Monte Carlo sampling with a goal of ensuring reliable
Monte Carlo experimental results.  The approach is based on a
multivariate central limit theorem for any finite combination of means
and quantiles assuming the underlying process is strongly mixing at
the appropriate rate.  We emphasize the IID and Markov chain sampling
cases.  However, the proof technique used indicates that the result
will hold more broadly since all that was required is a strong law for
sample means, a CLT for sample means, and a Bahadur quantile
representation.

Given the multivariate CLT, we provide a fast algorithm for
constructing hyperrectangular confidence regions having the desired
simultaneous coverage probability and a convenient marginal
interpretation.  The methods are then incorporated into standard ways
of visualizing the results of Monte Carlo experiments enabling the
practitioner to more easily assess the reliability of the results.

\section*{Acknowledgements}
The authors thank Vladimir Minin for helpful conversations about
assessing Monte Carlo error in Bayesian analyses.  We also thank the anonymous
associate editor and referees whose comments helped improve this paper.

\bigskip
\begin{center}
{\large\bf SUPPLEMENTARY MATERIAL}
\end{center}

\begin{description}
\item[PDF:] File \texttt{supp\_visual.pdf} contains a description of the Gibbs sampler used in Section~\ref{sec:bayes} and additional plots that were omitted to save space.
\item[R Code:] Contains \texttt{R} code that reproduces the simulations and plots along with a file \texttt{README.txt} with directions (.zip file).

\end{description}

\singlespacing
\bibliographystyle{apalike}
\bibliography{ref}

\end{document}